\documentclass{llncs}
\usepackage{amsfonts,multicol}
\usepackage{amsmath}

\usepackage{ifpdf}
\ifpdf
    \usepackage[pdftex]{graphicx}
\else
    \usepackage[dvips]{graphicx}
\fi

\usepackage{algorithm}
\usepackage{algorithmic}

        {\hspace*{\fill}$\Box$\par}

\newcommand{\junk}[1]{}

\DeclareMathAlphabet{\mathpzc}{OT1}{pzc}{m}{it}

\numberwithin{equation}{section}

\newcommand{\rev}{\mathcal{R}}

\newcommand{\ignore}[1]{}
\newcommand{\shortvers}[1]{}

\begin{document}

\title{Quasi-Proportional Mechanisms: Prior-free Revenue Maximization}
\author{Vahab Mirrokni
\thanks{Google Research, New York, NY. {\tt mirrokni@google.com}}
and
S. Muthukrishnan
\thanks{Google Research, New York, NY. {\tt muthu@google.com}}
and
Uri Nadav
\thanks{Tel-aviv University {\tt uri.nadav@gmail.com}. Part
of this work was done when the author was at Google.} }

\institute{}

\maketitle

\begin{abstract}
Inspired by Internet ad auction applications, we study the problem
of allocating a single item via an auction when bidders place very
different values on the item. We formulate this as the problem of
{\em prior-free} auction and focus on designing a simple mechanism
that always allocates the item. Rather than designing sophisticated
pricing methods like prior literature, we design better allocation
methods. In particular, we propose {\em quasi-proportional}
allocation methods in which the probability that an item is
allocated to a bidder depends (quasi-proportionally) on the bids.

We prove that corresponding games for both all-pay and winners-pay
quasi-proportional mechanisms admit pure Nash equilibria and this
equilibrium is unique. We also give an algorithm to compute this
equilibrium in polynomial time. Further, we show that the revenue of
the auctioneer is promisingly high compared to the ultimate, i.e.,
the highest value of any of the bidders, and show bounds on the
revenue of equilibria both analytically, as well as using
experiments for specific quasi-proportional functions. This is the
first known revenue analysis for these natural mechanisms (including
the special case of proportional mechanism which is common in
network resource allocation problems).
\end{abstract}
\thispagestyle{empty}

\newpage
\setcounter{page}{1}

\section{Introduction}

Consider the following motivating example. There is a single item (in
our case, an ad slot) to be sold by auction.
We have two bidders $A$ and $B$, $A$ with valuation $b_A=100$ and $B$
with valuation $b_B=1$.
Who should we {\em allocate} the item and
what is the {\em price} we charge? In the equilibrium of
the {\em first price} auction, $A$ wins by bidding $1+\epsilon$. We
(the auctioneer) get {\em revenue} of $1+\epsilon$
for some small $\epsilon>0$.
In the {\em second price} auction, $A$ wins and pays
$b_B+\epsilon=1+\epsilon$ for some
$\epsilon>0$ and the revenue is $1+\epsilon$, equivalent
to the first price revenue. So, neither generates revenue anywhere
close to the maximum valuation
of $\max\{b_A,b_B\}=100$.  Is there a mechanism
that will extract revenue close to the maximum valuation of bidders in
equilibrium?
What is the formal way to address this situation where
valuations are vastly different? In this paper,
we look at this problem in a general setting of prior-free auction
design, and study
revenue maximization. Further, we
propose a class of natural allocations and analyze them for revenue
and equilibrium
properties under different pricing methods.

Our motivation arises from allocation of ad slots on the Internet.
Consider the example of sponsored search where when a user
enters a phrase in a search engine, an auction is run among
advertisers who target that phrase to determine which ads will be
shown
to the user. There are several instances where the underlying {\em
value} is vastly different for the different participating
advertisers. For example, the phrase ``shoes'' may be targeted by both
high end as well as low end shoe retailers and may
have vastly different values, budgets or margins in their business.
Thus their bids will likely be vastly different.
In another example, we have display advertising, where users who visit
certain web sites are shown ``display'' ads like images, banners
or even video. Then, depending on the history of the user --- e.g.,
someone who is new to the website versus one who has been previously
--- different display advertisers value the user significantly
differently, and therefore their bid values will be vastly different.
In both these motivating scenarios, there are other issues to model
and this paper is not a study of these applications, but
rather, a study of a fundamental abstract problem inherent in these
applications.

\paragraph{Prior-free Auctions and Revenue Maximization.}
Revenue Maximization is a central issue in mechanism design and
has been studied extensively. A standard way for maximizing
revenue is to derive some value profile from the bids, calculate
bidder-specific reserve price, and run a second price
auction~\cite{M,BV,S}.
In the example above, say both buyers' {\em value} comes from some
random distribution.
Then, if we know this distribution, we can calculate a {\em reserve price}
$r$ using this distribution, and run a second-price auction with this
reserve price $r$, i.e, allocate the item to the highest bidder
$A$ and charge $A$
the $\max\{b_B,r\}$ if $b_A \geq r$ (else, the item remains unsold);
here, $b_A$ and $b_B$'s are bids by $A$ and $B$ resp.
Many such mechanisms are known; these mechanisms are incentive-compatible
(that is, each bidder has no incentive to lie), and even additionally
revenue-optimal,
perhaps as the number of bidders goes to infinity. Such methods that
rely on some assumptions over the values of bidders, i.e, that the
values are drawn from some distribution (known or unknown), are
called {\em prior-aware} mechanisms. Prior-aware mechanisms are
popular in Economics.
Still, from mathematical and practical point, the following questions arise:

\medskip
{\noindent \bf 1.} Are there prior-free mechanisms that work
independent of the value distributions of bidders?

This question is of inherent interest: what can be accomplished
without knowledge of the value distributions. This is also a question
that is motivated by practice. In practical applications, a way to use
prior-aware mechanisms is to rely on running the same auction
many times, and then use the history of bids to ``machine learn'' the
values. Of course in practice the parameters of the
auction change (users evolve), there is sparse data (query phrases
are rare), advertisers strategize in complex ways and their
values change over time (as they learn their own business feedbacks
better), or worse, even if the machine learning methods converge, they
provide ``approximate'' value distributions and we need to understand
the mechanisms under approximate distributions. As a result,
there are challenges in applying prior-aware mechanisms in practice
and a natural question is if they can be avoided.

\medskip
{\noindent \bf 2.} Are there prior-free mechanisms that work without reserve-prices?

This question is a more nuanced concern. First, when there is a
reserve price, the item may remain unsold in instances when $b_A < r$.
This may not be desirable in general. For example, in display ads, if
an ad slot is unsold, the webpage has to find a different template
without that ad slot or fill in that space with backup ads. Also, when
the item is not sold, the outcome is not
{\em efficient}, since the value to the advertisers (defined the value
of time to the winner) is not maximized. And in an ever more
nuanced note, advertisers do not find it transparent when the
mechanism has bidder-specific reserve prices, and often see it as a
bias.
This is more so when each advertiser may get many different
bidder-specific reserve prices corresponding to different search
phrases or
display ad locations as implied by the general prior-aware mechanisms
above. More discussions on mechanisms
that always assign the item can be found in~\cite{yishaypaper}.

{\em Prior-free} revenue-maximizing  mechanisms have been developed for various auction settings~\cite{Fiat02,HK,shanghua}.
Lower bounds show that prior-free truthful auction cannot achieve
revenue comparable to the revenue-optimal auctions with prior~\cite{Fiat02,HK,shanghua},
and the mechanism in~\cite{shanghua} achieves the best
possible revenue among prior-free truthful mechanisms.
Still, these mechanisms work by setting reserve prices, and do not address the second concern above.

\medskip
{\noindent \bf Our Contribution.}
We study a simple, practical prior-free mechanism that always allocates the item. In contrast to the approaches described above that
allocate the item to the highest bidder, but determine nontrivial prices, we focus on the allocation problem and allocate
the item {\em probabilistically}. Our contributions are as follows.

{\noindent \bf 1.} We propose a {\em quasi-proportional} allocation scheme where the probability that a bidder wins the
item depends (quasi-proportionally) on the bids.

As an example, for two bidders with bids $b_A$ and $b_B$, we allocate the item to bidder A with
probability  $\sqrt{b_A}\over \sqrt{b_A} + \sqrt{b_B}$, and to $B$ otherwise.
More generally in the presence of $n$ bidders with bid vector
$(b_1, \ldots, b_n)$, we consider a continuous and concave function $w$, and set the probability of
winning for bidder $i$ to $w(b_i) \over \sum_{1\le j\le n} w(b_j)$.
Thus the winner of the auction is not necessarily the bidder with the highest bid. The
special case when  $w(b_i)=b_i$ is known as the proportional allocation scheme and has been studied previously e.g., in \cite{JT,Kelly97,hajek}.
We study both payment methods that are common in auction theory, namely, {\em all-pay} (where all bidders pay
their bid no matter if they win the item or not) as well
as the {\em winner-pay} (only the winner pays her bid to the auctioneer) methods.

{\noindent \bf 2.} We study Nash equilibria of quasi-proportional mechanisms.

{\bf 2.1.} We prove that the corresponding games for both all-pay and winners-pay quasi-proportional
mechanisms admit pure Nash equilibria and this equilibrium is unique. We also give an algorithm to compute this equilibrium
in polynomial time.

{\bf 2.2.} We show that the revenue of the auctioneer is promisingly high, while
not losing much in the efficiency of the allocation.
More precisely, we compare the revenue of such mechanisms against the ultimate: $\max_i v_i$, the highest value of any
of the bidders, and show bounds on the revenue of equilibria in such mechanisms.
For example, consider an auction among two bidders with values $v_A=\alpha$ and $v_B=1$ respectively.
The revenue of equilibria for both first-price and second-price auctions
approaches $1$. Instead, with quasi-proportional mechanisms,
(i) for the all-pay mechanism with function $w(x) = x^{\gamma}$ where $\gamma\le 1$, the revenue of equilibrium is $\gamma \alpha^{1-\gamma}$,
and  (ii) for winners-pay mechanism,  where $\alpha >> 1$,  we show that the revenue of all-pay and winners-pay
mechanisms with functions $w(x) = x$ and $w(x) = \sqrt{x}$ are $\Omega(\alpha^{1\over 2})$
and $\Omega(\alpha^{2\over 3})$ respectively.
For the case of more than two bidders, we first show preliminary results for the revenue of various (specific) valuation vectors
for the case that the number of buyers tends to $\infty$, and then we present numerical results
for the revenue of equilibria for some key example functions such as $w(x) = \sqrt{x}$ and $w(x) = x$.
Taken together, these results give a set of analytical and experimental tools to bound the revenue of
these mechanisms against the $max_i v_i$ benchmark.

Proportional allocation, a special case of our quasi-proportional allocation,
has been studied extensively in literature, in particular for efficiency analysis. But even for this rather natural
allocation method, we do not know of any prior work on revenue analysis. \junk{UN:redundant We do not know of any prior work on
revenue analysis of quasi-proportional allocations, and our work contributes to this basic interest.}

\section{Preliminaries}\label{sec:prelim}
Consider a sealed-bid auction of a single item for a set $A=\{1,\ldots, n\}$ of  $n$ potential buyers.
Let the value of these $n$ buyers for the single item be $v_1\ge v_2\ge \cdots \ge v_n$.
Throughout this paper, we assume that $v_1=\alpha \ge 1$, and $v_n = 1$. Consider a concave function $w: R\rightarrow R$ (e.g., $w(x) = \sqrt{x}$).
Each buyer $i\in A$ bids an amount $b_i$ to get the item. A {\em quasi-proportional mechanism} allocates the item in a probabilistic manner. In particular,
the item is allocated to exactly one of the buyers, and the probability that buyer $i$ gets the item is $w(b_i) \over \sum_{j\in A} w(b_j)$.
For a bid vector $(b_1, \ldots, b_n)$, let $b_{-i}$ be the bid vector excluding the bid
of buyer $i$.
We study the following two variants of quasi-proportional mechanisms with two payment schemes.

\begin{enumerate}
\item {\bf All-pay Quasi-proportional Mechanisms.} The allocation rule in this mechanism is described above. For the payment scheme in this mechanism, each buyer pays her bid (no matter if he receives the item or not). This mechanism is ex-ante individually rational, but not ex-post individually rational.
Given the above payment scheme, in the all-pay mechanism,
we can write the utility of buyer $i$, as a function
of the bids vector as follows:
$$u_i(b) =u_i(b_i, b_{-i}) = v_i \frac{w(b_i)}{\sum_{j\in A} w(b_j)} - b_i .$$

\item {\bf Winners-pay Quasi-proportional Mechanisms.} The allocation rule in this mechanism is described above. For the payment scheme in this mechanism, the buyer who receives the item pays her bid, and the other buyers pay zero. This mechanism is ex-post individually rational.
As a result buyer $i$'s utility as a function of the bids is $$u_i(b) = u_i(b_i, b_{-i}) = \frac{w(b_i)}{\sum_{j\in A} w(b_j)}(v_i  - b_i).$$
\end{enumerate}

We are interested in Nash equilibria\footnote{Throughout this paper,
we study {\em pure} Nash equilibria and not mixed Nash equilibria.}
of the above mechanisms. We consider Nash equilibria of normal-form
games with complete information. In the corresponding normal-form
game of the quasi-proportional mechanism, the strategy of each buyer
$i$ is her bid. Formally, a bid vector $(b^*_1, \ldots, b^*_n)$ is a
{\em Nash equilibrium} if for any buyer $i$ and any bid $b'_i$, we
have $u_i(b^*)= u(b^*_i, b^*_{-i}) \ge u_i(b'_i, b^*_{-i})$.

In addition, we study efficiency and revenue of quasi-proportional mechanisms:
(i) the {\em efficiency} of a bid vector $(b_1, \ldots, b_n)$ is the expected valuation
of buyers, i.e., $\sum_{i\in A} (v_i{w(b_i) \over \sum_{j\in A} w(b_j)})$,
and (ii) the {\em revenue} of a bid vector $(b_1, \ldots, b_n)$ is the expected revenue
for the auctioneer given this bid vector,
i.e., $\sum_{i\in A} b_i$, in the all-pay auction, and $\sum_{i\in A} (b_i{w(b_i) \over \sum_{j\in A} w(b_i)})$, in the winner-pay auction.

\section{All-pay quasi-proportional mechanism: A warm-up example}\label{warm-up}
To demonstrate the kind of analyses we do, and to develop the intuition, we present a study
of revenue properties of an all-pay quasi-proportional mechanism for two buyers for functions $w(t) = t^\gamma$ where $\gamma\in [0,1]$ is a parameter. Let the bid of the first buyer be $y=b_1$ and the bid of the second buyer  $x=b_2$. As mentioned earlier, we assume $v_1=\alpha$ and $v_2=1$ are the valuations
of the two buyers.
The expected utility of the second buyer  is $ \frac{x^\gamma}{x^\gamma+y^\gamma} - x$,
and the utility of the first buyer  is  $\alpha \frac{y^\gamma}{x^\gamma+y^\gamma} - y$.

For a fixed $y$, the second buyer's  utility is a concave and increasing function of his bid, in the region $[0,\infty)$
and similarly, for a fixed $x$, the first buyer's utility is concave and increasing in his bid. Hence, in equilibrium,
both buyers have their first derivative nullified:
$\frac{\partial}{\partial x} \left( \frac{x^\gamma}{x^\gamma+y^\gamma} - x  \right) =  0$, and
$\frac{\partial}{\partial y} \left( \alpha \frac{y^\gamma}{x^\gamma+y^\gamma} - y \right) = 0.$ Thus, we get that

\begin{eqnarray*}
\frac{\gamma (x)^{\gamma-1} y^\gamma }{(x^\gamma+y^\gamma)^2} = 1 \mbox{\ \ \ \ \ }  \mbox{and \ \ \ \ \ }  \frac{\alpha \gamma (y)^{\gamma-1} x^\gamma}{(x^\gamma+y^\gamma)^2} = 1
\end{eqnarray*}

From which it follows that in equilibrium $\frac{y}{x} = \alpha$. Now, combining with the second equality, we get that
$\frac{\alpha \gamma (\alpha x)^{\gamma-1} x^\gamma}{((1+\alpha^\gamma)x^\gamma)^2} = 1$ or
$\frac{\alpha \gamma (\alpha x)^{\gamma-1} }{(1+\alpha^\gamma)^2 x^\gamma} = 1$,
and we get that $x=\frac{\gamma \alpha^\gamma}{(1+\alpha^\gamma)^2};   y=\alpha\frac{\gamma \alpha^\gamma}{(1+\alpha^\gamma)^2}.$
Hence, $$x+y = (1+\alpha)\frac{\gamma \alpha^\gamma}{(1+\alpha^\gamma)^2} \stackrel{\alpha\rightarrow \infty}{\longrightarrow} \gamma \alpha^{1-\gamma}.$$

Moreover, as ${y\over x} = \alpha$, the probability that buyer $2$ receives the item is $1\over 1+ \alpha^\gamma$, and
otherwise buyer $1$ gets the item. Thus, the efficiency of this mechanism is $1+\alpha^{\gamma+1} \over 1+ \alpha^{\gamma}$.
In particular, as $\alpha\rightarrow \infty$, the efficiency is arbitrarily close to $\alpha$. The most efficient
allocation rule is to assign the item to buyer $1$, and get efficiency $\alpha$.
%that the efficiency of equilibria compared to
That completes the analysis and shows that

\begin{theorem}
The all-pay quasi-proportional mechanism with two buyers guarantees a total revenue of
$(1+\alpha)\frac{\gamma \alpha^\gamma}{(1+\alpha^\gamma)^2}$ and
expected efficiency of $1+\alpha^{\gamma+1}\over 1+\alpha^{\gamma}$ in equilibrium.
In particular, for a large enough $\alpha$, the revenue is $\gamma \alpha^{1-\gamma}$
and efficiency is arbitrarily close to $\alpha$.
\end{theorem}

\iffalse
For example, in a second price auction~\cite{V, C, G} we are guaranteed to have a revenue of $1$ (and so in a first price auction, in equilibrium, by the revenue equivalence theorem). Here, if we choose $\gamma=1/2$, we get a revenue of $\frac{1}{2}\sqrt{\alpha}$,
which can be much more than $1$ for large $\alpha$. Choosing the optimal $\gamma$ would strictly
require the knowledge of $\alpha$, but the example above shows that an oblivious choice of $\gamma$ still gets better
revenue than second price auction  (when $\alpha >4$).
One can do better if for example, an upper bound $\alpha^*$ is provided for $\alpha$. Then, choose $\gamma= 1/\log \alpha^*$
and the revenue will be $\alpha/\log \alpha^*$ which is larger than $1$ for good upper bounds.
\fi

\section{Equilibrium: Existence and Uniqueness}\label{eql-compute}
In this section, we establish the existence and uniqueness of Nash
equilibria of both the all-pay and winners-pay quasi-proportional auctions.

\iffalse
Rosen~\cite{rosen}  defines the class of {\em concave} games, where each buyer's
strategy set is a convex and compact set, and where the utility of
every buyer  is a concave function in her own strategy.
Rosen~\cite{rosen} shows that a Nash equilibrium exists in every such
game.
%It is straightforward to verify that both the all-pay auction
%and the winner-pay auction belong to the class of concave games, and
%therefore a Nash equilibrium exists in these auctions.
Even Dar et. al\cite{EMN09} study the class of socially concave
games, defined as follows:
\fi

\begin{definition} [from \cite{EMN09}]\label{socially_concave} A game is {\em socially
concave} if the following holds:
\begin{enumerate}\label{theenumi}
    \item \label{concave_in_sum} There exists a strict convex combination
    of the utility functions which is a concave function. Formally, there
    exists an $n$-tuple $(\lambda_i)_{i\in A}$ , $\lambda_i>0$, and
    $\sum_{i\in A}\lambda_i = 1$, such that $g(x)=\sum_{i\in A}\lambda_i
    u_i(x)$ is a concave function in $x$.
    \item\label{convex_in_others}
    The utility function of each buyer  $i$, is convex in the
    actions of the other buyers. {\it I.e.,} for every $s_i\in S_i$
    the function $u_i(s_i, x_{-i})$ is convex in $x_{-i}\in S_{-i}$, where $S_i$ is the strategy space of agent $i$, and $S_{-i}=\prod_{j\in A, j\neq i} S_j$.
\end{enumerate}
\end{definition}

Rosen~\cite{rosen} defined the diagonal concavity property for concave games, and
showed that when it holds, the Nash equilibrium of the game is
unique.
Even Dar et al~\cite{EMN09} showed that if one of the
properties \ref{concave_in_sum} and \ref{convex_in_others} holds
with strict concavity or convexity, respectively, then the diagonal
concavity property holds.
Now, we show that a quasi-proportional auction is a socially concave
game. The uniqueness of Nash equilibrium would follow as a corollary
of~\cite{rosen} and~\cite{EMN09}.

\begin{lemma}\label{l:allpay_is_sc}
Let $\Gamma = (A, \{u_i\}_{i\in A}))$ be an all-pay
quasi-proportional auction, with utility functions for buyer $i$,
$u_i()$ defined as above and assume that the weight function $w()$
is a concave function, and that the strategy of each buyer is
restricted to a compact set $[B_{\min}, B_{\max}]$, where $0 <
B_{\min} < B_{\max} < \infty$. Then $\Gamma$ is a socially-concave
game.
\end{lemma}
\begin{proof}
To show property~\ref{concave_in_sum} holds, consider the weighted
social welfare function $g()$, where the utility of buyer $i$ is
weighted by $\lambda_i = \sum_{j\in A}{v_j}/v_i$.

\begin{eqnarray}\nonumber g(b) \equiv\sum_{i\in A}\lambda_i
u_i(b)= \sum_{i\in A} \frac{\sum_{j\in
A}{v_j}}{v_i}\left(  v_i \frac{w(b_i)}{\sum_{j\in A}
w(b_j)} - b_i \right) \nonumber & \\ \label{all_pay_sum} = \sum_{j\in A}{v_j} \sum_i
\frac{w(b_i)}{\sum_{j\in A} w(b_j)} - \sum_{i\in A} \frac{\sum_{j\in
A}{v_j}}{v_i}b_i = \sum_{j\in A}{v_j} - \sum_{i\in A} \frac{\sum_{j\in A}{v_j}}{v_i}b_i & \mbox{\ \ }
\end{eqnarray}

The first term in \ref{all_pay_sum} is a constant and the second
term is linear in $b$. Thus, $g()$ is a concave (linear) function of
$b$. To show that property \ref{convex_in_others} holds, we first
fix a buyer $i$ and an action $b_i\in [B_{\min},B_{\max}]$. Now,
consider the utility of buyer $i$ as a function of the actions
$b_{-i}$, when buyer $i$'s action is $b_i$:
$$u_i^{(b_i)}(b_{-i}) \equiv u_i(b_i,b_{-i}) = v_i\frac{w(b_i)}{w(b_i)  + \sum_{j\neq i} w(b_j)} - b_i$$
To show that $u_i^{(b_i)}(b_{-i})$ is a convex function, it suffices
to show that the function $$w(b_{-i}) = \frac{c}{c  + \sum_{j\neq i}
w(b_j)}, $$ is a convex function of the vector $b_{-i}$. Let
$g(b_{-i})=\sum_{j\neq i} w(b_j)$. The function $g(b_{-i})$ is
concave in $b_{-i}$ as it is the sum of $n-1$ concave functions. Let
$h(z)\equiv \frac{c}{c + z}$. The function $h()$ is convex in $z$
and decreasing in $\mathbb{R_+}$. The function $w=h(g(b_{-i}))$ is a
convex function as a composition of a convex decreasing function
with a concave function: For every two vectors $b^0,b^1$ and
$t\in[0,1]$ we have that $g(t b^0+(1-t)b^1)\geq t g(b^0)+(1-t)
g(b^1)$, by the concavity of $g()$. Consequently, $h(g(t
b^0+(1-t)b^1))\leq h(t g(b^0)+(1-t) g(b^1)) \leq t
h(g(b^0))+(1-t)h(g(b^1)$, where the first inequality follows from
the fact $h(z)$ is decreasing  when $z>0$, and the second inequality
follows the convexity of $h$.
\end{proof}

A similar lemma holds for winner-pay auctions, with weight function
of the form $w(x)=x^\gamma$, where $0<\gamma\leq 1$.

\begin{lemma}
Let $\Gamma = (A, \{u_i\}_{i\in A}))$ be an winner-pay
quasi-proportional auction, with utility functions for user $i$,
$u_i()$ defined as above and assume that the weight function
$w(x)=x^\gamma$, where $0<\gamma\leq 1$, and that the strategy of
each user is restricted to a compact set $[B_{\min}, B_{\max}]$,
where $0 < B_{\min} < B_{\max} < \infty$. Then $\Gamma$ is a
socially concave game.
\end{lemma}

\begin{proof}
Without loss of generality, we can restrict user $i$ bids to be less
than his value $v_i$. To show property \ref{convex_in_others} holds,
we notice that for every fixed $0<b_i<v_i$, we have that $v_i-b_i>0$,
and consequently, user $i$ utility function
$$u_i^{(b_i)}(b_{-i})=\frac{w(b_i)}{w(b_i)  + \sum_{j\neq i} w(b_j)} (v_i- b_i) = c \frac{w(b_i)}{w(b_i)  + \sum_{j\neq i} w(b_j)},$$
for some $c>0$. The last term was shown to be a convex function in
the proof of Lemma \ref{l:allpay_is_sc}.

To show property \ref{concave_in_sum} holds, we use the same weight
vector as in the proof of Lemma \ref{l:allpay_is_sc}. In this case,
the weighted social welfare is not an affine function, but it can be
shown that it is a convex function when $w(x)=x^\gamma$, when
$0<\gamma<1$.
\end{proof}

\section{Revenue of Quasi-proportional Mechanisms}\label{revenue}
In section~\ref{warm-up}, we computed the revenue of all-pay quasi-proportional mechanisms for two buyers,
and functions $w(x) = x^\gamma$. In this section, we first observe general properties for the revenue of equilibria of
quasi-proportional mechanisms. Then, we focus on two special functions and prove tight bounds
on the revenue of the winners-pay mechanisms.
The utility function $u_i(b_i,b_{-i})$ for both all-pay and winners-pay mechanisms is a
strictly concave function of $b_i$ in the region $[0,\infty]$ (as it is a
concave function minus a convex function).
As a result, in an all-pay quasi-proportional auction, we have:
$\frac{\partial}{\partial b_i} \left( {w(b_i) \over \sum_{i\in A}{w(b_i)}}v_i - b_i  \right) =  0$.
For a bid vector, $(b_1, b_2, \ldots, b_n)$, let $\sigma(b)=\sum_{i\in A} w(b_i)$.
When clear from context, we let $\sigma = \sigma(b)$.
As a result, in equilibrium,

\begin{eqnarray}\label{eq1}
\frac{\partial}{\partial b_i} \left( {w(b_i) \over \sigma}v_i - b_i  \right) =  0 \Rightarrow \mbox{\ \ \ }
%\frac{d}{d b_i} \left( {w(b_i) \over (t-w(b_i))+w(b_i)}v_i - b_i  \right) &=&  0 \Rightarrow \nonumber\\
( {w'(b_i)(\sigma-w(b_i)) \over \sigma^2})v_i = 1 \Rightarrow \mbox{\ \ \ }
v_i =  {\sigma^2 \over w'(b_i)(\sigma-w(b_i))}
\end{eqnarray}

Similarly, for winners-pay quasi-proportional mechanisms,
the bid of each buyer $i$ satisfies the following:

\begin{eqnarray}
\frac{\partial}{\partial b_i} \left( {w(b_i) (v_i - b_i)\over (\sigma-w(b_i)) + w(b_i)} \right) = {w'(b_i) (\sigma-w(b_i))\over \sigma^2}(v_i-b_i) - {w(b_i) \sigma\over \sigma^2} =0 \Rightarrow \nonumber
\end{eqnarray}
\begin{eqnarray}
% = 0 \nonumber
v_i =b_i+ {w(b_i) \sigma \over w'(b_i) (\sigma-w(b_i))} \label{eq2}
\end{eqnarray}

We will use equations~\ref{eq1} and ~\ref{eq2} in studying the revenue of the equilibrium for
various functions.
In both equations ~\ref{eq1} and \ref{eq2} for increasing concave functions like $w(x) = \sqrt{x}$, $v_i$ increases
as $b_i$ increases, i.e, fixing $b_{-i}$ $v_i$ is  monotonically increasing in terms of $b_i$. This observation
leads to the following fact:  For increasing and concave functions $w$, if
$v_1\ge v_2 \ge \ldots \ge v_n$, in the equilibrium bid vector $(b^*_1, b^*_2, \ldots, b^*_n)$,
we have $b^*_1\ge b^*_2\ge \ldots b^*_n$.

\subsection{Revenue for Winners-pay: Two Bidders.}
Here, we study winners-pay proportional mechanism for $w(x) = x$.
The utility of bidder $i$ as a function of the bids is
$$u_i(b) = \frac{b_i}{\sum_{j\in A} b_j}\left(v_i -
  b_i\right).$$
%\begin{equation}
%  u_i(b) =\left\{ \begin{array}{cc}  \frac{b_i}{\sum_{j\in A} b_j}\left(v_i -
%  b_i\right) & w\neq 0 \\
%    0 & w=0
%\end{array} \right.
%\end{equation}

Given this utility function, it is easy to see that
for $v_i>0$, in equilibrium $b_i>0$.
Let's fix $b_{-i}\neq 0$. In equilibrium, for every
$i$ with bid $b_i>0$,

$$\frac{\partial}{\partial b_i}u_i(b_i,b_{-i}) = -\frac{b_i(v_i-b_i)}{(\sum_{j\in
A} b_j)^2} + \frac{v_i-b_i}{\sum_{j\in A} b_j} -
\frac{b_i}{\sum_{j\in A} b_j} = -1 + \frac{(\sum_{j\neq i}
b_j)(\sum_{j\neq i} b_j + v_i)}{(\sum_{j\in A} b_j)^2}$$

and we get that in equilibrium, \begin{equation}\label{eq:1} b_i =
\sqrt{(\sum_{j\neq i} b_j)(\sum_{j\neq i} b_j + v_i)}  - \sum_{j\neq
i} b_j,  \text{ for  every } i\in A.\end{equation}

The revenue from the proportional mechanism as described above is

\begin{equation}\sum_{i\in A} \Pr[{\text agent  \ i \ wins}]\cdot b_i=\sum_{i\in
A} \frac{b_i}{\sum_{j\in A} b_j}b_i = \frac{\sum_{i\in A}
b_i^2}{\sum_{j\in A} b_j}\end{equation}

Consider a setting of two buyers with values $v_1, v_2$. We can,
without loss of generality, assume that $v_2=1$.

\begin{theorem}\label{clm:winnerpay_proportional}
  In the case of two buyers, the revenue from the winners-pay proportional
  mechanism is $O(\sqrt{\alpha})$, where
  $\alpha=\max(v_1,v_2)$. Moreover, for arbitrarily large $\alpha$,
  the efficiency of this mechanism is arbitrarily close to $\alpha$.
\end{theorem}
\begin{proof}
Let $\rev(v_1,v_2)$ denote the expected revenue of the auctioneer, when the values of the agents 1, and 2 are $v_1,v_2$ respectively. We assume without loss of generality that $\alpha=v_1>v_2=1$. In this case,
notice that the revenue is dominated by the bid of the first buyer:

\begin{equation}
\rev(v_1,v_2)=\frac{b_1^2+b_2^2}{b_1+b_2} = b_1+b_2 -
\frac{2b_1b_2}{b_1+b_2} \geq b_1+b_2 - 2
\end{equation} where the last inequality follows the observation that in equilibrium,
$b_2\leq 1$. Hence, to get a lower bound on the revenue, it suffices
to have a lower bound on $b_1$.

Next, we show that in equilibrium, $1/3 \leq b_2 \leq 1/2$:
First, note that using~\ref{eq:1}, we can easily show that for $v_1=v_2=1$,
the equilibrium is $b_1=b_2=1/3$.
In the equilibrium, $b_2 = \sqrt{b_1^2+b_1} - b_1$. Consider the function
$w(x)=\sqrt{x^2+x}-x$. It is increasing in $[0,\infty)$ as
$$w'(x)=-1+\frac{1+2x}{2\sqrt{x^2+x}}=\frac{-2\sqrt{x^2+x}+1+2x}{2\sqrt{x^2+x}} >0,
$$ where the last inequality follows from
$$1+2x=\sqrt{(1+2x)^2} = \sqrt{1+4x+4x^2}>\sqrt{4x^2+4x}=2\sqrt{x^2+x}.$$

As a result, if we fix $v_2=1$, then in equilibrium, when $v_1\geq
v_2$, we get $b_1\geq 1/3$. Therefore, $b_2$ which increases as a
function of $b_1$ is always greater than $1/3$ in equilibrium. Also,
since $\lim_{x\rightarrow\infty}w(x)\rightarrow 1/2$.
Finally, we get {$b_1 = \sqrt{b_2^2+b_2 v_1}-b_2 \geq \sqrt{1/9+v_1/3} - 1/3 $},
which proves the theorem. The claim about efficiency
of the mechanism follows from the fact that $b_1$ tends to $\infty$ as
$\alpha$ tends to $\infty$.
\end{proof}

A similar technique can be used for showing a lower bound on the
revenue in quasi-proportional winner-pay auctions, with weight
function $w(x)=\sqrt{x}$, which asymptotically yields a higher
revenue. The proof is left to the appendix.

\begin{theorem}\label{thm:sqrt}
The revenue from the winners-pay mechanism for two bidders, with weight function $w(x)=\sqrt{x}$ is $O(\alpha^{2/3})$, where
$\alpha=\max(v_1,v_2)$. Moreover, for arbitrarily large $\alpha$,
the efficiency of this mechanism is arbitrarily close to $\alpha$.
\end{theorem}

We will give numerical results for revenue of other settings
like $w(x) = x^{1/4}$ in Section~\ref{simulations}.
As for other functions, we prove an upper bound on the revenue
of both all-pay and winners-pay mechanisms for
$w(x) = \log(x+1)$, and show that the revenue is
not more than $\alpha\over \log(\alpha)$ (See Section~\ref{otherfunctions}).

\subsection{Revenue for many buyers.}\label{manybuyers}
%Computing the exact revenue bounds for $n$ bidders for general valuation vectors is hard.
Here, we analyze the
revenue for two special valuation vectors for $n$ bidders, i.e, (i) uniform valuation vector, $v_i=V$, and
(ii) valuation vector $v_1=\alpha$, and for $i\not = 1$, $v_i = 1$ for $i\in A$.
The second type of valuation is important as it captures examples in which there is a large
gap between the highest valuation and value of other buyers.

\begin{theorem}
For the uniform valuation vector where $v_i=V$ for all $i\in A$, the revenue in the equilibrium for function
$w(x) = x^{\gamma}$ is ${n-1\over n}\gamma V$ for all-pay mechanism, and is
$V({1\over 1+({n\over n-1})\gamma})$ for winners-pay mechanism.
Moreover, the equilibrium revenue for uniform valuation vector for
function $w(x) = \log(x+1)$ for both all-pay and winners-pay mechanisms
is asymptotically $V\over \log V$ as $V,n\rightarrow \infty$ .
\end{theorem}
\begin{proof}
Proofs follows directly from Equations ~\ref{eq1} and \ref{eq2} by noting that
$b_i=b$, $t=nw(b)$. Also we use the fact that, for all-pay mechanism, the revenue is $nw(b_i)$, and
for winners-pay, the revenue is $w(b)$.
\end{proof}

\begin{theorem}\label{thm:manyplayers}
For the valuation vector $(\alpha, 1, 1,\ldots, 1)$, the revenue in the equilibrium of winners-pay quasi-proportional
mechanism converges to a constant as $n$ goes to $\infty$ for a fixed $\alpha$. Moreover the revenue of all-pay quasi-proportional
mechanism for function $w(x) = x^{\gamma}$ goes to zero as $n$ goes to $\infty$ for a fixed $\alpha$.
\end{theorem}

The above theorem shows some bounds on the revenue for a fixed $\alpha$ and as $n$ tends to $\infty$.
It would be interesting to understand the trade-off between the revenue for large $\alpha$ and $n$.
In particular, it would be interesting to compute the revenue for a fixed $n$ as $\alpha$ tends to $\infty$.

\section{An Efficient Algorithm and Numerical Study}\label{simulations}
In this section, we present an efficient algorithm for computing Nash equilibria
of quasi-proportional mechanisms and then using this algorithm, we present a family
of plots showing the quality of the mechanisms.

\subsection{A polynomial-time algorithm for equilibrium computation}\label{computation}
In \cite{EMN09}, Even Dar et. al. describe a natural process that
converges to a Nash equilibrium in every socially concave game. This
method is useful for computing Nash equilibrium of the all-pay and
winner-pay auctions. The process considered is known as
\emph{no-regret} dynamics. Informally, a buyer's update process is
said to have no-regret, if in the long-run, it attains an average
utility which is not significantly worse than that of the best fixed
action in hindsight (in the context of auctions, the best fixed
bid). Even Dar et. al. show that if every buyer uses an update process
with no-regret property, in a repeated socially concave game,
the joint average action profile converges to a Nash Equilibrium.
Many efficient algorithms for attaining the no-regret property (also
known as no-external-regret), exist~\cite{Zinkevich,Hazan1,Kalai}. In order to compute a Nash equilibrium of the
all-pay auction, and the winner-pay auction, one could simulate the
process of running a no-regret algorithm for every buyer that
participates in the auction. The rate at which the average vector of
bids converges to Nash equilibrium, depends on the vector $\lambda$,
which existence is guaranteed in property \ref{concave_in_sum}. In
particular, there exists no-regret algorithms (e.g.,
\cite{Zinkevich}), such that the rate of convergence to Nash
equilibrium, for the quasi-proportional mechanisms, is
$O(\frac{n}{\sqrt{t}}\frac{\sum_{j\in A}v_j}{v_{\min}})$, (I.e., at
time $t$ of the simulation process, the average bids vector is an
$\epsilon^t$-Nash equilibrium, where $\epsilon^t
=O(\frac{n}{\sqrt{t}}\frac{\sum_{j\in A}v_j}{v_{\min}})$. Algorithm
\ref{alg:computeNE} describes the simulation of running simultaneous
no-regret for every buyer, where the actual no-regret algorithm used
is GIGA~\cite{Zinkevich}.

\begin{algorithm}\caption{Algorithm for computing Nash equilibrium bids for
the quasi-proportional auction.}
\textbf{Input:} a vector $v=\{v_1,v_2,\ldots,v_n\}$.\\
\textbf{Output:} an $\epsilon$-Nash Equilibrium, $b_1,\ldots,b_n$.\\
\begin{algorithmic}
    \STATE $b^0 \leftarrow(1,1,\ldots,1)$
    \FOR{$t=1$ to $T=O(\frac{n}{\epsilon}\frac{\sum_{j\in A}v_j}{v_{\min}})$}
        \FORALL{ $i\in A$}
                \STATE $y_i^t \leftarrow b_i^{t-1} + \frac{1}{\sqrt{t}}\frac{\partial}{\partial b_i} u_i(b)$
                \IF {$y_i^t > v_i$}
                    \STATE {$b_i^t\leftarrow v_i$}
                \ELSE
                    \STATE{$b_i^t \leftarrow \max(y_i^t,0)$}
                \ENDIF
        \ENDFOR
    \ENDFOR
    \STATE {\bf return} $b^T$
\end{algorithmic}\label{alg:computeNE}
\end{algorithm}

%\begin{table}[here]\label{alg:computeNE}
%
% \abovedisplayskip=6pt
% \belowdisplayskip=6pt
%
%\hrule \vskip 6pt \textbf{Algorithm: ComputeNE\quad}
%
%\vskip 3pt \textbf{Input:} a vector $v=\{v_1,v_2,\ldots,v_n\}$.
%
%\vskip 3pt \textbf{Output:} an $\epsilon$-NE, $b_1,\ldots,b_n$.
%
%\begin{enumerate}%\label=(\roman*}
%    \item Set $b^0 \leftarrow(1,1\ldots,1)$
%    \item Set $t\leftarrow 1$
%    \item For $t=1$ to $T=O(\frac{n}{\epsilon}\frac{\sum_{j\in
%    A}v_j}{v_{\min}})$:
%    \begin{enumerate}%[label=\roman*]
%        \item For each $i\in A$
%            \begin{enumerate}
%                \item $y_i^t \leftarrow b_i^{t-1} + \frac{1}{\sqrt{t}}\frac{\partial}{\partial b_i} u_i(b)$
%                \item if $y_i^t > v_i$, then $b_i^t\leftarrow v_i$,
%                otherwise $b_i^t \leftarrow \max(y_i^t,0)$
%            \end{enumerate}
%
%       \item $t\leftarrow t+1$
%    \end{enumerate}
%    \item return $b^T$
%\end{enumerate}
%\hrule \caption{Algorithm for computing NE bids for the
%quai-proportional auction.}
%\end{table}

\subsection{Numerical Revenue Computation}
In this section, we present numerical results for the revenue of the
all-pay and winners-pay quasi-proportional auctions with different
weight functions and different number of buyers. Figures
\ref{fig1}-\ref{fig4} describe the revenue as a function of the
highest value  for the item, over all the bidders, denoted by $\alpha$.
Figure \ref{fig1} describes the revenue in an all-pay auction with
two bidders --- one bidder has a `high' value $\alpha \geq 1$, and
the other bidder has a value of 1. We consider two versions of the all-pay
auctions. In the first, we used a weight function
$w(z)=\sqrt{z}$, and in the second we used a weight function
$w(z)=z^{1\over 4}$.
Next, in Figure \ref{fig2}, we consider the same setting as in
Figure \ref{fig1}, for the winners-pay auction. The revenue in
equilibrium is presented for three different versions of the
winners-pay auction: The lowest curve describes the winner pay
auction with the linear weight function $w(z)=z$. The middle curve
describes the revenue when the weight function is $w(z)=\sqrt{z}$
and the upper curve describes the revenue when the weight function
is $w(z)=z^{1/4}$.

In Figures \ref{fig3}, and \ref{fig4} we study numerically the
revenue in a winners-pay auction when the number of bidders varies
from $n=2$ to $n=5$. The bidders' private values are such that a
single bidder has a high value $\alpha\geq 1$, and the other $n-1$
bidders have a low value of 1. Each curve in Figures
\ref{fig3},\ref{fig4} describes the revenue in equilibrium as a
function of $\alpha$, and each different curve corresponds to a
different number $n$ of bidders. Figure \ref{fig3} and \ref{fig4}
differ in the weight function used: in Figure \ref{fig3} we used
$w(z)=z$, and in Figure \ref{fig4} we used $w(z)=\sqrt{z}$. In
Theorem \ref{thm:manyplayers}, we show that the revenue in a
winners-pay auction, with values profile $(\alpha,1,1,\ldots,1)$
asymptotically goes to 0, as the number of bidders with value 1
tends to $\infty$. It is interesting however to notice that in both
Figures \ref{fig3} and \ref{fig4}, while the number of bidders is
kept relatively small, the revenue actually increases with the
number of low-value bidders.

\section{Concluding Remarks}

We study a natural class of quasi-proportional allocation mechanisms.
Combined with all-pay or winner-pay methods, this gives a simple
prior-free auction
mechanism without any reserve prices. Our analytical and experimental
study shows the revenue under various quasi-proportional functions in
equilibrium, and
we showed existence of a unique Nash equilibrium that can also be
computed efficiently. We believe quasi-proportional mechanisms will
find applications and a deeper
understanding of their properties will be useful.
%In addition, we
%list a few open directions.
%
%{\noindent \bf Constant-factor approximation for prior-free revenue
%maximization.}
An interesting open question is to design
an auction for a single item that achieves a total revenue of
constant factor of $\alpha=\max_i v_i$ in equilibria.
We proved that simple quasi-proportional mechanisms show promising revenue
properties in equilibria, however none of our mechanisms achieve a
constant approximation
factor of $\alpha$ (off by at least facor $\log \alpha$). A main open
problem is to design a mechanism for a single item that
achieves a constant factor of $\alpha$ in equilibria while not losing
much in the efficiency
of the allocation.
%
%\junk{
%
%{\noindent \bf Revenue Maximization for many buyers.}
Also as we discussed in Section~\ref{manybuyers}, the promising revenue properties
of quasi-proportional mechanisms for small number of buyers disappears as the
number of buyers tends to $\infty$. An interesting open question is to
modify the mechanism to ensure good revenue
properties when many buyers are in the system. A simple idea is that for any number of bidders, the auctioneer runs a quasi-proportional mechanism
among the highest two bids.
One hopes such mechanisms have good revenue properties,
however, we can show that such mechanisms may not admit any pure Nash
equilibria.
%}

\junk{

{\noindent \bf Revenue properties for small number of buyers.}
We analyzed the revenue for many buyers in Section~\ref{manybuyers} and reported
numerical results for the auctions in the presence of a small number
of buyers. However, we do not fully understand the tradeoff  between the
number of buyers and the gap between the maximum valuation and other valuations
for the revenue of quasi-proportional mechanisms. In particular, we do not
know the equilibria revenue  for large $n$ in settings where
$\alpha$ tends to $\infty$.
%Moreover, for more than two buyers,
%we only study equilibria of specific valuation vectors.
We leave a more complete study of revenue and efficiency of equilibria
for quasi-proportional mechanisms
as an interesting open question.
}

\junk{
{\bf \noindent Acknowledgements.} We thank Amos Fiat and Yishay Mansour for interesting
discussions, and pointing out related work.
}

\begin{figure}
\centering
\includegraphics[width=0.6\textwidth]{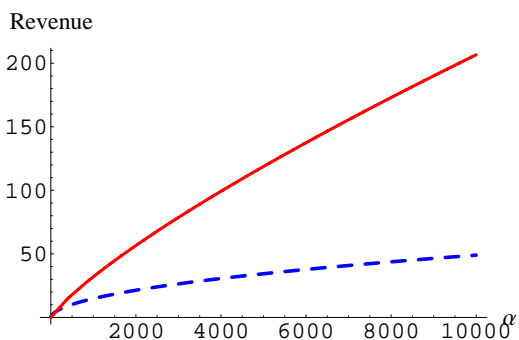}
\caption{\small{Revenue from equilibrium bids in an all-pay auction
with two bidders with values $\alpha$, and 1 respectively. The lower
curve describes an all-pay auction with weight function
$w(x)=\sqrt{x}$. The upper curve describes an all-pay auction with
weight function $w(x) = x^{1/4}$.}} \label{fig1}
\end{figure}

\begin{figure}
\centering
\includegraphics[width=0.6\textwidth]{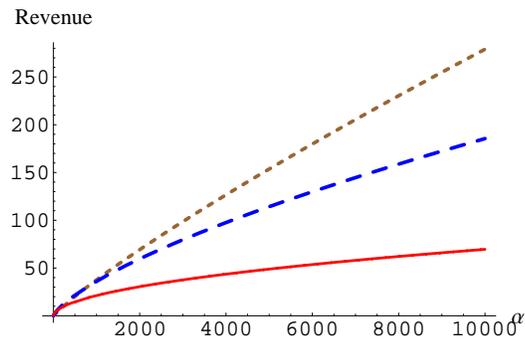}
\caption{\small{Revenue from equilibrium bids in a winners-pay
auction with two bidders with values $\alpha$, and 1 respectively.
The lower, middle, and upper  curves describes a winners-pay auction
with weight functions $w(x)=x$, $w(x)=\sqrt{x}$, and $w(x)=x^{1/4}$
respectively.}} \label{fig2}
\end{figure}

\junk{

%\iffalse
\begin{figure}[here]
  \begin{center}
\mbox{\includegraphics[width=0.5\textwidth]{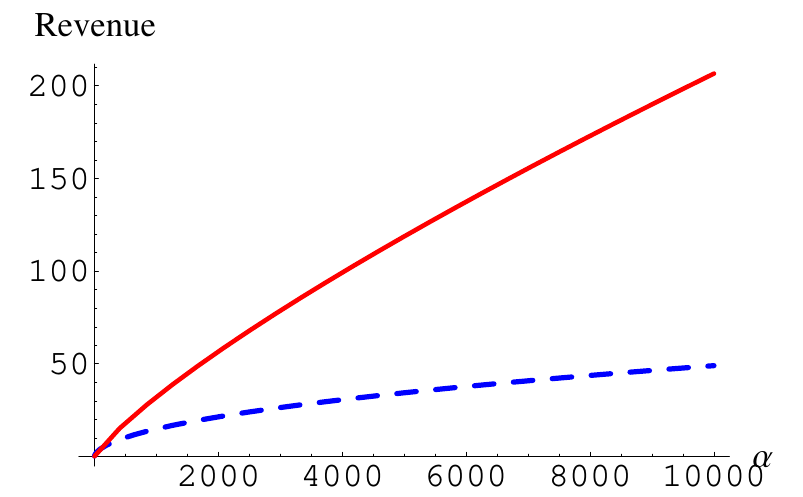}}
  \end{center}
 \hfill
\caption{\small{Revenue from equilibrium bids in an all-pay auction
with two bidders with values $\alpha$, and 1 respectively. The lower
curve describes an all-pay auction with weight function
$w(x)=\sqrt{x}$. The upper curve describes an all-pay auction with
weight function $w(x) = x^{1/4}$.}} \label{fig1}
\end{figure}

\begin{figure}[here]
  \begin{center}
\mbox{\includegraphics[width=0.5\textwidth]{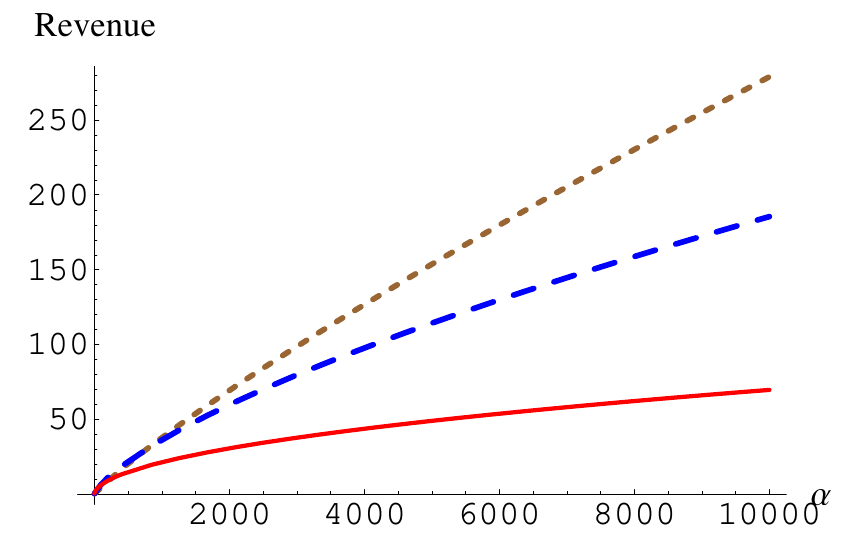}}
  \end{center}
 \caption{\small{Revenue from equilibrium bids in a winners-pay auction with two
bidders with values $\alpha$, and 1 respectively. The lower, middle,
and upper  curves describes a winners-pay auction with weight
functions $w(x)=x$, $w(x)=\sqrt{x}$, and $w(x)=x^{1/4}$
respectively.}} \label{fig2}
\end{figure}
}

\begin{figure}
\centering \centering
\includegraphics[width=0.6\textwidth]{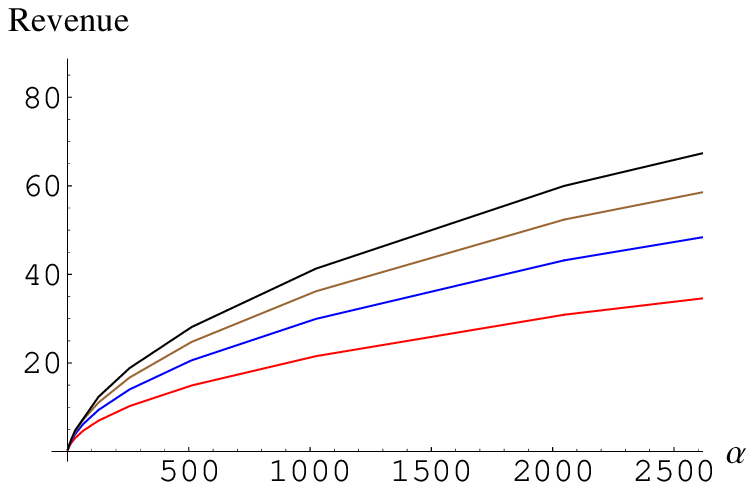}
\caption{\small{A winners-pay auction with weight $w(x)=x$, and
value profile $(\alpha,1,\ldots,1)$. The curves from lowest to
highest describe the revenue when the number of bidders with value 1
is 1,2,3,4 respectively}} \label{fig3}
\end{figure}

\begin{figure}
\centering
\includegraphics[width=0.6\textwidth]{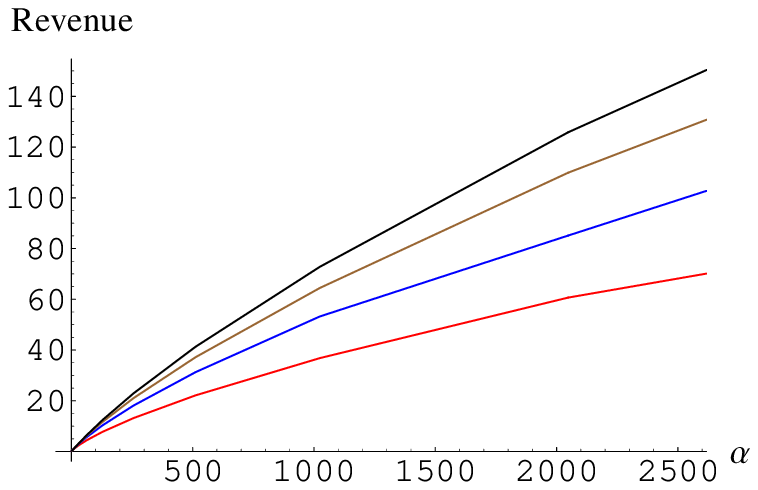}
\caption{\small{A winners-pay auction with weight $w(x)=\sqrt{x}$,
and value profile $(\alpha,1,\ldots,1)$. The curves from lowest to
highest describe the revenue when the number of bidders with value 1
is 1,2,3,4 respectively.}} \label{fig4}
\end{figure}

\junk{
\begin{figure}[here]
  \begin{center}
\mbox{\includegraphics[width=0.5\textwidth]{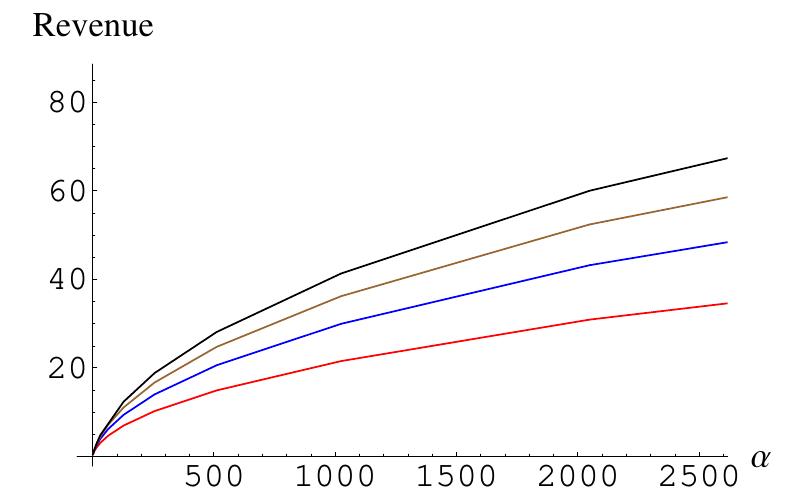}}
  \end{center}
\caption{\small{A winners-pay auction with weight $w(x)=x$, and value
profile $(\alpha,1,\ldots,1)$. The curves from lowest to highest
describe the revenue when the number of bidders with value 1 is
1,2,3,4 respectively}} \label{fig3}
\end{figure}

\begin{figure}[here]
  \begin{center}
\mbox{\includegraphics[width=0.5\textwidth]{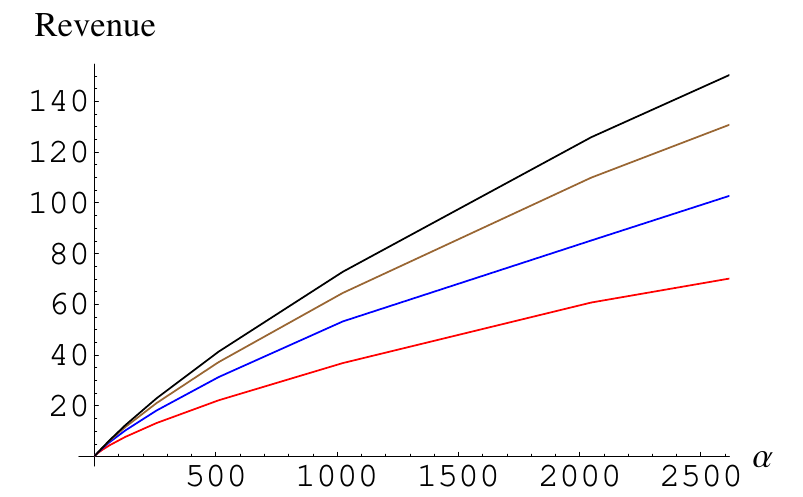}}
  \end{center}
\caption{\small{A winners-pay auction with weight $w(x)=\sqrt{x}$,
and value profile $(\alpha,1,\ldots,1)$. The curves from lowest to
highest describe the revenue when the number of bidders with value 1
is 1,2,3,4 respectively.}} \label{fig4}
\end{figure}
}

\newpage

\bibliographystyle{abbrv}

\begin{thebibliography}{10}

\bibitem{Hazan1}
J. Abernethy, E. Hazan and A. Rakhlin.
\newblock {\em Competing in the Dark: An Efficient Algorithm for Bandit Linear Optimization}. In COLT 2008.


\bibitem{BV}
Baliga and Vohra.
\newblock Market Research and Market Design.
\newblock \url{http://www.kellogg.northwestern.edu/faculty/baliga/htm/mrandmd.pdf}, 2003.



\bibitem{BKV}
M. Baye and D. Kovenock and C. de Vried.
\newblock The all-pay auction with Complete Information
\newblock Economic Theory, 8, 291-305.



\bibitem{CG}
Y. Che and I. Gale.
\newblock Expected revenue of all-pay auctions and first-price sealed-bid auctions with budget constraints.
\newblock Economic Letters, 1996, 373-379.

\bibitem{C}
E.~Clarke.
\newblock Multipart pricing of public goods.
\newblock {\em Public Choice}, 11:17--33, 1971.


\bibitem{EMN09} E. Even Dar, Y. Mansour, U. Nadav.
\newblock Convergence in Proportional Games.
\newblock STOC 2009.




\bibitem{Fiat02} Amos Fiat, Andrew V. Goldberg, Jason D. Hartline, and Anna R. Karlin, \newblock
Competitive generalized auctions, \newblock STOC 2002,  pages 72-81.



\bibitem{G}
T.~Groves.
\newblock Incentives in teams.
\newblock {\em \em Econometrica}, 41(4):617--631, 1973.

\bibitem{hajek}
B.~Hajek and G.~Gopalakrishnan.
\newblock Do greedy autonomous systems make for a sensible internet?, 2002.
\newblock presented at the Conference on Stochastic Networks, Stanford
  University.

\bibitem{HK}
J. Hartline, A. Karline, \newblock
Profit Maximization in Mechanism Design,
\newblock In Algorithmic Game Theory, Editors: Noam Nisan, Tim Roughgarden, Eva Tardos, and Vijay Vizarani, October 2007.


\bibitem{JT}
R. Johari and J.N. Tsitsiklis.
\newblock Efficiency loss in a network resource allocation game.
\newblock Mathematics of Operations Research, 29(3):407–435, 2004.


\bibitem{Kalai}
A Kalai and S. Vempala.
\newblock {\em  Efficient algorithms for online decision problems}. In J. Comput. Syst. Sci. 71(3): 291-307, 2005.

\bibitem{Kelly97}
F.~Kelly.
\newblock Charging and rate control for elastic traffic.
\newblock {\em European Transactions on Telecommunications}, 8:33--37, 1997.


\bibitem{yishaypaper} De Liu  and Jianqing Chen
\newblock
Designing online auctions with past performance information,
\newblock Decision Support Systems, 42 (2006) 1307–1320.

\bibitem{shanghua}
Pinyan Lu, Shang-Hua Teng, Changyuan Yu,
\newblock Truthful Auctions with Optimal Profit.
\newblock  WINE 2006: 27-36


\bibitem{M}
R.\ Myerson.
\newblock Optimal auction design.
\newblock {\em Mathematics of Operations Research}, 6:58--73, 1981.

\bibitem{rosen}
J. Rosen.
\newblock Existence and uniqueness of equilibrium points for concave n-person games,
\newblock {\em Econometrica}, 520-534, 1965.


\bibitem{S}
I. Segal.
\newblock Optimal Pricing Mechanisms with Unknown Demand.
\newblock {\em American Economic Review}, 93 (3) June 2003, pp.509 - 529.
\newblock {\em http://www.stanford.edu/~isegal/pricing.pdf}.


\bibitem{V}
W.~Vickrey.
\newblock Counterspeculation, auctions and competitive-sealed tenders.
\newblock {\em Finance}, 16(1):8--37, 1961.

\bibitem{Zinkevich} M. Zinkevich.
\newblock Online convex programming and generalized
infinitesimal gradient ascent.
\newblock {\em Twentieth International Conference
on Machine Learning}, 2003.


\end{thebibliography}

\newpage

\appendix
\section{Proof Sketch of Theorem~\ref{thm:sqrt}}
\begin{proof}
The proof follows a similar line to that of Theorem
\ref{clm:winnerpay_proportional}. Here, we give a proof
sketch. We start by taking the derivative
of the utility function,

\begin{eqnarray*}\frac{\partial}{\partial b_i}u_1(b_1,b_2) &=& -
\frac{v_1 -b_1}{2(\sqrt{b_1}+\sqrt{b_2})^2} +
\frac{v_1-b_1}{2(\sqrt{b_1}+\sqrt{b_2})\sqrt{b_1}} -
\frac{\sqrt{b_1}}{\sqrt{b_1}+\sqrt{b_2}}\\\frac{\partial}{\partial
b_i}u_2(b_1,b_2) &=& - \frac{v_2 -b_2}{2(\sqrt{b_1}+\sqrt{b_2})^2} +
\frac{v_2-b_2}{2(\sqrt{b_1}+\sqrt{b_2})\sqrt{b_2}} -
\frac{\sqrt{b_2}}{\sqrt{b_1}+\sqrt{b_2}}\end{eqnarray*}
Assigning $b_1\equiv t_1^2$ and $b_2\equiv t_2^2$, and re-ordering
we get that in equilibrium:

\begin{eqnarray*}
  t_1 = \frac{2t_2^3}{-v_1+3t_2^2} & t_2 = \frac{2t_1^3}{-v_2+3t_1^2}
\end{eqnarray*}

Assuming that the lower value is $v_1=1$, the proof continues by
showing a lower bound on $t_1$, which in return is used for showing
that $t_2 = O(v_2^{1/3})$.
\end{proof}

\section{Revenue for other functions.}\label{otherfunctions}
\begin{theorem}
For the function $w(x) = \log(x+1)$ and for both all-pay and winners-pay mechanisms, the revenue
of the equilibrium for two buyers with values $1,\alpha$ where $\alpha>>1$
is at most  $c\alpha\over \log^2\alpha$ for a constant $c$.
\end{theorem}
\begin{proof}
Using Equation~\ref{eq1} for all-pay mechanism, we get that for $w(x) = \log(x+1)$,
$v_i={\sigma^2\over (\sigma-w(b_i))}w(b_i)$. Let the equilibrium be $(b^*_1, b^*_2)$.
Thus, for the first buyer,
$$\alpha = {(b^*_1+1)[\log(b^*_1+1)+\log(b^*_2+1)]^2\over \log(b^*_2+1)}.$$

Since, $b^*_2\le v_2=1$, we know that $\log(b^*_2+1)\le 1$, and thus
$$\alpha \ge {(b^*_1+1)\log^2(b^*_1+1)\over \log(b^*_2+1)}\ge (b^*_1+1)\log^2(b^*_1+1).$$
Therefore, for a constant $c$, $b^*_1\le {c\alpha\over \log^2 \alpha}$.

Similarly, for winners-pay mechanism, using Equation~\ref{eq2} for $w(x)=\log(x+1)$,
we observe that

$$\alpha = b^*_1 + {\log{b^*_1+1}[\log(b^*_1+1) + \log(b^*_2+1)](b^*_1+1)\over \log(b^*_2+1)}.$$
And since $b^*_2\le v_2=1$, we can easily show that $b^*_1\le {c\alpha\over \log^2 \alpha}$.
\end{proof}

\begin{remark}
Using function $w(x) = \log(\log(x+1)+1)$, one can prove an upper bound for
the revenue
of $\alpha\over \log(\alpha) \log\log^2(\alpha)$. More generally, let $w$ be the function
of $k$ consecutive application of $\log$ function (i.e, $w(x) =  \log(\ldots(\log(\log(x+1)+1)+1)+1)$),
be denoted by $w(x) = \log^{(k)}(x)$, then
we can show that the revenue is at most
$$\alpha\over \log(\alpha) \log\log(\alpha)
\ldots f^{(k-2)}(\alpha) f^{(k-2)}(\alpha) f^{(k)})^2 (\alpha).$$
%In fact, this revenue is less than the revenue extracted from mechanism
\end{remark}

\section{Proof Sketch of Theorem~\ref{thm:manyplayers}}
\begin{proof}
We analyze the winners-pay mechanism first. Because of the valuation vector,
we know that $b_i=b_j$ for $i,j\not =1$, we can set $B=b_1$ and
$b=b_i$ for $i\not = 1$. Therefore, using Equation~\ref{eq2} and the fact that
$\sigma = \sigma(b) = w(B) + (n-1)w(b)$, we get
\begin{eqnarray*}
\alpha & = & b_i+ {w(b_i) \sigma \over w'(b_i) (\sigma-w(b_i))} \\
1 & = & b + {w(b) (w(B) + (n-1)w(b)) \over w'(b_i) (w(B) + (n-1)w(b))}
\end{eqnarray*}

Substituting, $w(x)=x^\gamma$ and $w'(x) = \gamma x^{\gamma-1}$, as $n\rightarrow \infty$, the second equation above implies that
$b\cong {\gamma\over \gamma+1}$.

Using the above equations, by substituting one can show that as $\alpha$ is fixed and $n\rightarrow \infty$, $B=O(\alpha)$,
and thus the revenue from this mechanism is
$${w(B)B\over w(B) + (n-1)w(b)} + {w(B)B\over w(B)+ (n-1)w(b)} = \Theta(1).$$

Using equation~\ref{eq1} and by setting $B=b_1$ and $b=b_i$ for $i\not =1$,
one can perform similar computations and show that $b\cong {\gamma\over n}$,
and $B\cong {\alpha^{1\over \gamma} \over n}$ as $n\rightarrow \infty$.
As a result, the revenue of this mechanism tends to zero as $n$ tends to
$\infty$ and $\alpha$ is fixed.
\end{proof}

\end{document}